%% file: main.tex
\documentclass[acmsmall, nonacm]{acmart}

\AtBeginDocument{%
  }

\usepackage[utf8]{inputenc}
\usepackage{color}
\usepackage{graphicx,subcaption}
\usepackage{xspace}
\usepackage{hyperref}
\usepackage{enumitem}
\usepackage{cleveref}
\usepackage{lineno}
\usepackage{arydshln}
\usepackage{caption}
\newcommand{\ie}{\textit{i.e.}}

\renewcommand{\epsilon}{\varepsilon}
\newcommand\set[1]{\ensuremath{\{#1\}}}

\newcommand{\tup}[1]{{\langle #1 \rangle}}
\newcommand\vars{\textit{vars}}
\newcommand\key{\textit{key}}
\newcommand\certain{\textup{\textsc{certain}}}

\newcommand\conp{\textup{\sc coNP}\xspace}
\newcommand\ptime{\textup{\sc PTime}\xspace}
\newcommand\fo{\textup{FO}\xspace}

\newcommand\cert{\textup{Cert}}
\newcommand\Cqk{\cert_k(q)\xspace}
\newcommand\matching{\textup{\textsc{matching}}(q)\xspace}

\newcommand\bu{\bar u}
\newcommand\bv{\bar v}
\newcommand\sjf{\ensuremath{\textit{sjf}}}
\newcommand{\tripath}{\textsc{tripath}\xspace}
\newcommand\g{g}
\newcommand\bg{\bar g}
\newcommand{\bkey}{\overline{\key}}
\newcommand\rkey{\text{r-key}}
\newcommand{\adom}{\textit{adom}}
\newcommand{\twowaydet}{2way-determined\xspace}
\newcommand{\twowaydetcy}{2way-determinacy\xspace}
\newcommand{\clique}{\textit{clique}}

\begin{document}


\title{A Dichotomy in the Complexity of Consistent Query Answering for Two Atom
  Queries With Self-Join}


\author{Anantha Padmanabha}
\affiliation{%
  \institution{IIT Dharwad}
  \city{Dharwad}
  \country{India}
  }  
\email{ananthap@iitdh.ac.in}

\author{Luc Segoufin}
\affiliation{%
  \institution{INRIA, ENS-Paris, PSL University}
  \city{Paris}
  \country{France}
  }  
\email{luc.segoufin@inria.fr}

\author{Cristina Sirangelo}
\affiliation{%
  \institution{Université Paris Cité, CNRS, Inria, IRIF}
  \city{Paris}
  \country{France}
}
\email{cristina@irif.fr}

\renewcommand{\shortauthors}{Anantha Padmanabha, Luc Segoufin, \& Cristina Sirangelo}

\begin{abstract}
  We consider the dichotomy conjecture for consistent query answering under
  primary key constraints. It states that, for every fixed Boolean
  conjunctive query $q$, testing whether $q$ is certain (i.e.  whether it evaluates
  to true over all repairs of a given inconsistent database) is either
  \ptime or \conp-complete. This conjecture has been verified for
  self-join-free and path queries. We show that it also holds for
  queries with two atoms.
\end{abstract}




\maketitle

\section{Introduction}\label{sec-intro}
\input{introduction}

\section{Preliminaries}\label{sec-prelims}
\input{prelims}

\section{Dichotomy classification}\label{sec-classify}
\input{classify}

\section{First \conp-hard case}\label{sec-no-hmd}
\input{no-hmd}

\section{The greedy fixpoint algorithm}\label{sec-cert-algo}
\input{cert-algo}

\section{First polynomial time case}\label{sec-key subseteq ptime}
\input{key-subset}

\section{\twowaydet queries}\label{sec-tripath definition}\label{sec-nice tripath}
\input{sec-prime}

\section{Queries with no \tripath and \ptime}\label{sec-tripath or chase}
\input{tripath-or-chase}

\section{Fork-\tripath and \conp-hardness}\label{sec-fork tripath and conp}
\input{tripath-conp}

\section{Queries that admit only triangle-\tripath}\label{sec-triangle tripath}
\input{triangle-tripath}

\section{Conclusion}\label{sec-conclusion}
\input{conclusion}

\begin{acks}
 This work supported is by ANR QUID, grant ANR-18-CE40-0031. Padmanabha worked on part of this project during his affiliation with VALDA, ENS, Paris where he was supported by ANR-19-P3IA-0001 (PRAIRIE 3IA Institute).
\end{acks}

\bibliographystyle{ACM-Reference-Format}
\bibliography{biblio-repair}

\appendix

\end{document}

%% file: introduction.tex
A relational database often comes with integrity constraints. With the attempts to harness  data from complex environments like big data, social media etc., where the database is built by programs that go over a large data dump, more often than not we end up with a database that violates one or more of the integrity constraints. This is because in such heterogeneous sources, the data is often incomplete or ambiguous. 
Inconsistencies in databases also occur while integrating data from multiple sources.

To deal with this problem, one approach is to {\it clean} the database when it is being built and/or modified.
 However,  this task is not easy as it is inherently non-deterministic: there may be many equally good candidate tuples to add/delete/update to make the  database consistent. In the absence of additional information (which is often the case), these decisions are arbitrary.
 
There is another way to cope with this problem: we allow the database to be inconsistent and the problem is handled during query evaluation. In this approach, we retain the inconsistent database as it is and we rely on the notion of database {\em repair}. Intuitively repairing a database corresponds to obtaining a consistent database by making minimal changes to the inconsistent one. A conservative approach to evaluate a query is to evaluate it over every possible repair and retain only the \emph{certain answers}, \ie~ the query answers which are true on all repairs. This approach is called {\em consistent query answering}~\cite{DBLP:conf/pods/ArenasBC99, DBLP:conf/pods/Bertossi19}.

This approach for handling inconsistency has an advantage of not loosing any information and avoids making arbitrary choices to make the database consistent. However, since we need to evaluate the query on all the repairs, this will affect the complexity of query evaluation. The impact will of course depend on the type of integrity constraints and on the definition of a repair; but most often there could be exponentially many ways to repair a database.

 For a fixed boolean conjunctive query $q$, the decision version of the certain query answering problem is the  following: given an inconsistent database $D$ as input, does $q$ evaluate to true on all the repairs of $D$?

When we consider primary key constraints, the natural notion of a repair is to pick exactly one tuple for every primary key. Thus, every repair is a subset of the given inconsistent database. But there could be exponentially many repairs for the given database.
 For a fixed boolean conjunctive query, in the presence of primary key constraints, checking for certain answers is
 in \conp. This is because, to check that the query is not certainly true, it is enough to guess a subset of the database which forms a repair and verify that it makes the query false. However, there are queries for which the problem can be solved in \ptime and for some queries, the problem is \conp-hard.

The main conjecture  for consistent query answering in the presence of primary keys
is that there are no intermediate cases: for a fixed boolean conjunctive query $q$, the consistent query answering problem for $q$ is either
solvable in \ptime or \conp-complete.

The conjecture has been proved for self-join-free Boolean conjunctive
queries~\cite{DBLP:journals/tods/KoutrisW17} and path queries \cite{DBLP:conf/pods/KoutrisOW21}.
However, the conjecture remains open
for arbitrary conjunctive queries, in particular for queries having self-joins (\ie~ having at least two different atoms using the same relation symbol).

In this paper we show that the conjecture holds for conjunctive queries with
two atoms. As the case of self-join-free queries has already been solved \cite{DBLP:journals/ipl/KolaitisP12},
we consider only queries consisting of two atoms over the same relation symbol.

Towards proving the conjecture we start by introducing the notion of
\twowaydetcy and we distinguish two separate cases.  The first case proves the dichotomy for
all the two-atom queries with self-joins that are not \twowaydet; these are identified via
syntactic conditions. For these queries \conp-hardness is obtained through a
reduction from the self-join-free case with two
atoms~\cite{DBLP:journals/ipl/KolaitisP12}. On the other hand tractable cases are
obtained via the greedy fixpoint algorithm developed in~\cite{ICDTJournal} for
self-join free queries.

For queries that are \twowaydet we use a semantic characterization. To this end
we introduce the notion of \tripath, which is a database of a special form, and further
classify \tripath into triangle-\tripath and fork-\tripath.  For \twowaydet
queries the existence of a fork-\tripath establishes the desired complexity
dichotomy. In particular we prove that the certain answering problem is
\conp-hard for queries which admit a fork-\tripath while it is in \ptime otherwise. In the
latter case the polynomial time algorithm is a combination of the greedy
fixpoint algorithm of \cite{ICDTJournal} and a bipartite matching-based algorithm
(also introduced in \cite{ICDTJournal}).

Our second result further refines the polynomial time case by identifying
classes of queries for which the greedy fixpoint algorithm correctly computes certain answers (assuming
$\ptime\neq\conp$). The frontier turns out to be the presence of a
triangle-\tripath.  Indeed we show that for \twowaydet queries which do not admit a
\tripath at all (neither fork-\tripath, nor triangle-\tripath) the greedy
fixpoint algorithm alone is correct. Furthermore we also prove that this
algorithm fails to compute the certain answer to \twowaydet queries admitting a
triangle-\tripath.

\paragraph*{Related work}
The case of \emph{self-join-free} conjunctive queries with two atoms was
considered by Kolaitis and Pema~\cite{DBLP:journals/ipl/KolaitisP12}. Proving
the dichotomy in the presence of self-joins requires a completely different
technique. However we use the \conp-complete characterization
of~\cite{DBLP:journals/ipl/KolaitisP12} for solving a special case in our
analysis.

We rely heavily on the polynomial time algorithms developed
in~\cite{ICDTJournal}. In this work a simple greedy fixpoint algorithm, referred to as $\Cqk$, was introduced and was shown to solve all the \ptime 
cases of self-join-free conjunctive queries (and also path queries). Moreover \cite{ICDTJournal}  shows that some two-atom
queries with self-join which are in \ptime cannot be solved by $\Cqk$, but a different algorithm based on bipartite matching will work for them. We essentially show that a combination of these two
algorithms solve all the polynomial time cases of two-atom conjunctive
queries. In essence we show that if the combination of the two algorithms does not work then the query is \conp-hard.

We do not rely on the notion of attack graph or any other tools used for
self-join queries, developed by Koutris and
Wijsen~\cite{DBLP:journals/tods/KoutrisW17}.


%% file: prelims.tex
We consider boolean conjunctive queries over relational databases. As our queries
will have only two atoms and because the self-join-free case is already solved,
we can assume that these two atoms refer to the same relational symbol. Therefore we consider relational schema with only one relational symbol, associated with a primary key constraint.

A relational schema consist of a relation symbol $R$ with signature
$[k,l]$, where $k\ge 1$ denotes the arity of $R$ and the first $l~ (\ge 0)$ positions form the primary key of $R$.

We assume an infinite domain of {\bf elements} and an infinite
set of variables. A {\bf term} is of the form $R(\bar t)$ where
$\bar t$ is a tuple of elements or variables of arity $k$. A term $R(\bar t)$
is called a{\bf fact} if $\bar t$ is a tuple of elements, and $R(\bar t)$ is
called an {\bf atom} if $\bar t$ is a tuple of variables. We use $a,b,c$ etc to
denote facts and $A,B,C$ etc to denote atoms. 

Given a term $R(\bar t)$ we let $R(\bar t)[i]$ denote the variable / element at
$i$-th position of $\bar t$. For a set of positions $I$ we let
$R(\bar t)[I] = \{ R(\bar t)[i]\mid i\in I\}$. Let $S$ be the set of all $k$
positions of $R$. If $a$ is a fact then we write
$\adom(a)$ for $a[S]$. Similarly if $A$ is an atom
then we write $\vars(A)$ for $A[S]$. We define
the {\bf key} of $R(\bar t)$ to be the tuple $\bkey(R(\bar t))$ consisting of
the first $l$ elements of $\bar t$ and let
$\key(R(\bar t)) = R(\bar t)[K]$, where $K$ is the set of the first $l$
(key) positions of $R$.  For instance, if $R$ has
signature $[5,3]$ and $A=R(xyx~yz)$, we have $\bkey(A) = (x,y,x)$,
$\key(A) = \{x,y\}$ and $\vars(A) = \{x,y,z\}$.  Two terms $R(\bar t_1)$ and
$R(\bar t_2)$ are key-equal  if $\bkey(R(\bar t_1)) = \bkey(R(\bar t_2))$ and we denote it by $R(\bar t_1) \sim R(\bar t_2)$.

  A {\bf database} is a \emph{finite} set of facts. We
say that a database $D$ is of size $n$ if there are $n$ facts in $D$.  A
database $D$ is {\bf consistent} if it does not contain two
distinct key-equal facts. A {\bf block} in $D$ is a maximal subset
of $D$ that contains key-equal facts. A {\bf repair} of $D$ is a $\subseteq$-maximal
consistent subset of $D$. Note that $D$ can be partitioned into disjoint blocks and every repair picks exactly one fact from every block. If $r \subseteq D$ is a repair and $a$ is a fact in $D$ then for any $a'\sim a$ we denote as $r[a\to a']$ the repair  obtained from $r$ by replacing  $a$ by $a'$.

A {\bf query} $q$ is given by two atoms $A$ and $B$ and it corresponds to
the Boolean conjunctive query $\exists \bar y~ A \land B$ where $\bar y$ is the
tuple of all the variables in $\vars(A)\cup \vars(B)$. Since every variable is
quantified, we ignore the quantification and write $q=AB$. For instance, if the
query is $q = \exists xyzu~ R(xyx~uz) \land R(yxu~zu)$ then we let
$A = R(xyx~uz)$, $B=R(yxu~zu)$ and write $q = AB$.

A database $D$ satisfies a query $q=AB$, denoted by $D \models q$ (sometimes
denoted by $D\models AB$), if there exists a mapping $\mu$ from $\vars(A)\cup \vars(B)$ to elements such that $\mu(A),\mu(B)\in D$.
 In this case
the pair $(\mu(A),\mu(B))$  of (not necessarily distinct)  facts of $D$ is called
a {\bf solution} to $q$ in $D$. We also say that the fact $\mu(A)$ matches $A$ and $\mu(B)$ matches $B$. Different mappings give different solutions. The set of solutions to $q$ in $D$ is
denoted by {\bf $q(D)$}.  We will also write
{\bf $D \models q(ab)$} to denote that 
$(a,b)$ is a solution to $q$ in $D$ via some $\mu$.
We also write $D\models q\set{ab}$ to denote
$D\models q(ab)$ or $D\models q(ba)$. If $D$ is clear from the context we
simply write $q\set{ab},~q(ba)$ etc.

A query $q$ is {\bf certain} for
a database $D$ if all repairs of $D$ satisfy $q$. For a fixed  query\footnote{Along  standard lines, we adopt the
\emph{data complexity} point of view, \ie~ the query is fixed and we measure
the complexity as a function on the number $n$ of facts in $D$.} 
$q$, we denote by {\bf $\certain(q)$} the problem of determining, given
a database $D$, whether $q$ is certain for $D$. We write $D\models\certain(q)$  or $D\in \certain(q)$ to denote that $q$ is certain for $D$. 
Clearly the problem is in \conp as one can guess a (polynomial sized) repair $r$ of $D$ and verify that $r$
does not satisfy $q$.

We aim at proving the following result:

\begin{theorem}\label{main-theorem}
  For every (2-atom) query $q$, the problem $\certain(q)$ is either in \ptime
  or \conp-complete.
\end{theorem}

Note that $\certain(q)$ is trivial if $q$ has only one atom. As we deal with
data complexity, it is then also trivial for any query equivalent (\emph{over
  all consistent databases}) to a query with one atom.  For a query $q=AB$ this
can happen in two cases: $(1)$ there is a homomorphism from $A$ to $B$ or from
$B$ to $A$ ;~ $(2)$ $\bkey(A)=\bkey(B)$ (the query is then always equivalent
over consistent databases, to a single atom $R(C)$ where $C$ is the most general term that has homomorphism from both $A$ and $B$.) Hence, we will
assume in the rest of this paper that $q = AB$ is such that $\bkey(A)\ne \bkey(B)$
  and $q$ is not equivalent to any of its atoms.

  In the rest of the paper we will often underline the first $l$ positions of
  an atom or a fact in order to highlight the primary-key positions. We then
  write $R(\underline{xyz}~uv)$ to denote an
  atom involving a relation of signature $[5,3]$. Similarly we write
  $R(\underline{\alpha\beta\gamma}~\delta\epsilon)$ to denote a fact over the
  same signature.

%% file: classify.tex
The decision procedure for deciding whether the certainty of a query is
hard or easy to compute works as follows.

$\bullet$ We first associate to any query a canonical self-join-free query by
simply renaming the two relation symbols. If certainty of the resulting query
is hard, and this can be tested using the syntactic characterization
of~\cite{DBLP:journals/ipl/KolaitisP12}, then it is also hard for the initial
query. This is shown in \Cref{sec-no-hmd}.

$\bullet$ In \Cref{sec-key subseteq ptime} we give a simple syntactic condition guaranteeing
that the greedy polynomial time fixpoint algorithm of~\cite{ICDTJournal}
(presented in \Cref{sec-cert-algo}) computes certainty.

The remaining queries, where the two syntactic tests mentioned above fail, are called \twowaydet. They enjoy some nice
semantic properties that are described in \Cref{sec-tripath definition} and that we exploit to pinpoint their complexity.
Towards this, we define in \Cref{sec-tripath definition} a special kind of
database called \tripath whose solutions to the query have a particular
structure. We distinguish two variants of \tripath, as fork-\tripath and
triangle-\tripath.

$\bullet$ If the query does not admit any \tripath (\ie~ neither fork \tripath nor
triangle \tripath) then certainty can be computed using the greedy fixpoint algorithm as
show in  \Cref{sec-tripath or chase}.

$\bullet$ If the query admits a fork-\tripath, then certainty is \conp-hard as shown in
\Cref{sec-fork tripath and conp}.

$\bullet$ Finally, for queries that admits a triangle-\tripath but no
fork-\tripath, certainty can be computed in \ptime. This is shown in \Cref{sec-triangle
  tripath}. For such queries, we prove that the algorithm of
\Cref{sec-cert-algo} is not expressive enough to compute certainty. However we prove that a combination of it together with a known
bipartite matching-based algorithm (again from \cite{ICDTJournal}) is correct.


%% file: no-hmd.tex
Given a query $q = AB$, we can associate it with a canonical self-join-free query $\sjf(q)$ over a
schema with two distinct relational symbols\footnote{In \cref{sec-prelims} we have
  defined all the notions with respect to a single relation in the
  vocabulary. In this section, and only here, we consider two relations. Since the definitions are standard, we will not state them explicitly.} $R_1$ and $R_2$ of the same arity as
$R$. The query $\sjf(q)$ is defined by replacing $R$ by $R_1$ in $A$ and $R$
by $R_2$ in $B$. For instance, if
$q_1 = R(\underline{xu}~xv) \land R(\underline{vy}~uy)$ then
$\sjf(q_1) = R_1(\underline{xu}~xv) \land R_2(\underline{vy}~uy)$.  Intuitively,
$\sjf(q)$ is the same query as $q$ but with two different relation names.

We show that computing the certainty of $q$ is always harder than
computing the certainty of $\sjf(q)$. 
This is the only place where we use the assumption that $q$ is not equivalent to a one-atom query.

\begin{proposition}\label{prop-sjf-to-sj}
Let $q$ be a query. There is a polynomial time reduction from $\certain(\sjf(q))$
to $\certain(q)$.
\end{proposition}

\begin{proof}[Proof sketch.]
  Assume $q=AB$ where $A$ and $B$ are atoms using the relational symbol
  $R$. Let $R_1$ and $R_2$ be the symbols used in $\sjf(q)$.
  Let $D$ be a database containing $R_1$-facts and $R_2$-facts. We construct in
  polynomial time a database $D'$ containing $R$-facts such that
  $D \models \certain(\sjf(q))$ iff $D'\models \certain(q)$.

  For every fact $a=R_1(\bu)$ of $D$, let $\mu(a)=R(\bv)$ be a fact where every
  position $i$ of $\bv$ is the pair $\tup{z,\alpha}$ where $z$ is the variable
  at position $i$ in $A$ while $\alpha$ is the element at position $i$ in
  $\bu$. Similarly if $a=R_2(\bu)$ then $\mu(a)=R(\bv)$ where every position
  $i$ of $\bv$ is the pair $\tup{z,\alpha}$ where $z$ is the variable
  at position $i$ in $B$ while $\alpha$ is the element at position $i$ in
  $\bu$. Let $D'=\mu(D)$. It turns out that $D'$ has the desired property
  and this requires that $q$ is not equivalent to a one-atom query.
\end{proof}

It follows from \Cref{prop-sjf-to-sj} that whenever $\sjf(q)$ is \conp-hard
then $\certain(q)$ is also \conp-hard. It turns out that we know
from~\cite{DBLP:journals/ipl/KolaitisP12} which self-join-free queries with two
atoms are hard. This yields the following result.

\begin{theorem}\label{thm-sjf-to-sj}
  Let $q=AB$ be such that both the following conditions hold :
  \begin{enumerate}
    \item $\vars(A)\cap \vars(B) \not\subseteq \key(A)$ and $\vars(A)\cap
      \vars(B) \not\subseteq \key(B)$ and $\key(A)\not\subseteq \key(B)$ and $\key(B)\not\subseteq \key(A)$;\label{negHA}

  \item 
  $\key(A) \not\subseteq \vars(B)$ or 
  $\key(B) \not\subseteq \vars(A)$. \label{negHMD}
  \end{enumerate}
  Then $\certain(q)$ is \conp-complete.
\end{theorem}

For instance we can deduce from \Cref{thm-sjf-to-sj} that the query $q_1$
mentioned above is such that $\certain(q_1)$ is \conp-complete since $u$ and
$v$ are shared variables but $u \not\in \key(B)$, $v\not\in\key(A)$, moreover $\key(B) \not\subseteq \key(A)$ and $x\in \key(A)$ but is not in $\vars(B)$.

Note that the converse of \Cref{prop-sjf-to-sj}
is not true. For instance, the query
$q_2 = R(\underline{xu}~xy) \land R(\underline{uy}~xz)$ is such that
$\certain(\sjf(q_2))$ can be solved in polynomial time by the
characterization of~\cite{DBLP:journals/ipl/KolaitisP12}, but as we will see,
$\certain(q_2)$ is \conp-hard.

%% file: cert-algo.tex
In most of the cases where we prove that $\certain(q)$ is in \ptime, we use the following greedy fixpoint algorithm which was introduced in \cite{ICDTJournal}. 
For a fixed query $q$ and $k \geq 1$, we define the algorithm $\Cqk$. It takes
a database $D$ as input and runs in time $O(n^k)$ where $n$ is the size of
$D$. For a database $D$, a set $S$ of facts of $D$ is called a $k$-set if $|S|\le k$
and $S$ can be extended to a repair (\ie~ S contains at most one fact from
every block of $D$).

The algorithm inductively computes a set $\Delta_k(q,D)$ of $k$-sets while
maintaining the invariant that for every repair $r$ of $D$ and every
$S\in \Delta_k(q,D)$ if $S\subseteq r$ then $r\models q$. The algorithm returns
{\it yes} if eventually $\emptyset \in \Delta_k(q,D)$. Since all repairs
contain the empty set, from the invariant that is maintained, it follows that
$D\models \certain(q)$. The set $\Delta_k(q,D)$ is computed as follows:

Initially $\Delta_k(q,D)$ contains all $k$-set $S$ such that $S\models q$. Clearly,
this satisfies the invariant.
Now we iteratively add a $k$-set $S$ to $\Delta_k(q,D)$ if there exists a block
$B$ of $D$ such that for every fact $u\in B$ there exists
$S' \subseteq S\cup \{ u\}$ such that $S' \in \Delta_k(q,D)$. Again, it is
immediate to verify that the invariant is maintained.

This is an inflationary fixpoint algorithm and notice that the initial and
inductive steps can be expressed in \fo. If $n$ is the number of facts of $D$,
the fixpoint is reached in at most $n^k$ steps.

For a fixed $k$, we write {\bf $D \models \Cqk$} or {\bf $D \in \Cqk$} to denote that $\Cqk$ returns {\it yes} upon input $D$.
Note that $\Cqk$ is always an under-approximation of $\certain(q)$, \ie~
whenever $\Cqk$ returns  {\it yes} then $q$ is certain for the input database. However,
$\Cqk$ could give false negative answers. In \cite{ICDTJournal} it is proved that this algorithm captures all polynomial time cases for self-join-free queries and path queries by choosing $k$ to be the number of atoms in the query.


%% file: key-subset.tex
In view of \Cref{thm-sjf-to-sj}, 
it remains to consider the case where one of the conditions
of \Cref{thm-sjf-to-sj} is false.
In this section we prove that if the condition~(\ref{negHA}) is false for $q$ then
$\certain(q)$ is in \ptime.  By symmetry we only consider the case where  $\vars(A)\cap \vars(B) \subseteq \key(B)$ or
$\key(A)\not\subseteq \key(B)$. The other case follows from the fact that $q=AB$ is equivalent to the query
$BA$.

\begin{theorem}\label{thm-hp-X+}
  Let $q = AB$ be such that $\key(A)\subseteq \key(B)$ or
  $\vars(A) \cap \vars(B) \subseteq \key(B)$. Then $\certain(q) = \cert_2(q)$,
  hence $\certain(q)$ is in \ptime.
\end{theorem}

From \Cref{thm-hp-X+} it follows that the complexity of computing certain answers for queries like
$q_3=R(\underline{x}~ y)\land R(\underline{y}~ z)$ and
$q_4=R(\underline{xx}~ uv) \land R(\underline{xy}~ ux)$ is in \ptime. In the
case of $q_3$ this is because the only shared variable is $y$ and $\key(B)=\set{y}$. In the case of $q_4$ this is because $\key(A)=\set{x}\subseteq\set{xy}=\key(B)$.

In the remaining part of this section we prove \Cref{thm-hp-X+}.
The main consequence of the assumption on the query is the following zig-zag
property. We say that $q$ satisfies the zig-zag property if for all
database $D$, for all facts $a,b,b',c$ of $D$ such that $a\not\sim c$,
$a\neq b$ and $b\sim b'$, if $D\models q(ab)$ and $D \models q(cb')$ then
$D\models q(ab')$.

\begin{lemma}\label{lemma-hp}
  Let $q=AB$ be such that $\vars(A) \cap \vars(B) \subseteq \key(B)$ or
  $\key(A)\subseteq \key(B)$. Then $q$ satisfies the zig-zag property.
\end{lemma}

The key to the proof of \Cref{thm-hp-X+} is the following lemma.

\begin{lemma}\label{lemma-hp-X+}
  Let $q = AB$ be a query satisfying the zig-zag property. For all databases
  $D$ and for all repair $r$ of $D$ if $r \models q(ab)$ then
  $\set{a}\in \Delta_2(q,D)$ or there exists a repair $s$ of $D$ such that
  $q(s)\subsetneq q(r)$.
\end{lemma}

\begin{proof}[Proof sketch]
  Assume that $\set{a}\not\in \Delta_2(q,D)$. This means that there exists some
  $b'\sim b$ such that $\set{a,b'} \not\in \Delta_2(q,D)$. Then consider $r' =
  r[b\to b']$. Notice that the only new solutions in $r'$ that are not in $r$
  should involve $b'$. Hence if $b'$ is not a part of  any solution, then $r'$
  is the required repair. Otherwise note that $b'$ can only be a part of
  solutions of the form $r' \models q(b'c)$ (if $r'\models q(cb')$ for some $c$
  then by zig-zag property we we also have $r'\models q(ab')$ which is a
  contradiction). We also have $\set{b'}\not\in \Delta_2(q,D)$ (otherwise
  $\set{a,b'} \in \Delta_2(q,D)$ which is a contradiction) and we can repeat
  the construction. This way, we inductively build a sequence of repairs $r_0,r_1,\ldots r_n$ such that $r_0 = r$ and $r_n$ is the desired repair.
\end{proof}

\begin{proof}[Proof of \Cref{thm-hp-X+}]

Let $D\models \certain(q)$. We prove that $D\models\cert_2(q)$. 
  
We first show that every repair $r$ of $D$ contains some fact $a \in r$ such
that $\set{a} \in \Delta_2(q,D)$. Pick an arbitrary repair $r$ of $D$. 

Let $r'$ be
a minimal repair having $q(r') \subseteq q(r)$ (possibly $r'=r$). Since all the
repairs contain a solution, there exist facts $a, b$ of $r'$ such that
$r'\models q(ab)$. By Lemma~\ref{lemma-hp}, $q$ satisfies the zig-zag property,
thus we can apply Lemma~\ref{lemma-hp-X+}. This implies that
$\{a\}\in\Delta_2(q,D)$, otherwise one can construct another repair $s$ of $D$
such that $q(s) \subsetneq q(r')$, contradicting minimality of $r'$.  By the
choice of $r'$, one has also $r \models q(ab)$, thus we have $a\in r$ such
that $\{a\}\in\Delta_2(q,D)$.

Now let $r_{min}$ be a repair of $D$ containing the minimum number of facts $a$ such that $\set{a} \in \Delta_2(q,D)$. Let $m$ be this minimum number. By the property proved above there exists a fact $b \in r_{min}$ such that $\set{b} \in \Delta_2(q,D)$. We claim that for all $b' \sim b$, we have $\set{b'} \in \Delta_2(q,D)$. Suppose not, then $r_{min}[b \rightarrow b']$ contains $m-1$ facts in $\Delta_2(q,D)$, contradicting minimality of $r_{min}$.

Overall this proves $\emptyset \in \Delta_2(q,D)$ and hence $D\models \cert_2(q)$. 
\end{proof}

%% file: sec-prime.tex
From \Cref{thm-sjf-to-sj} and \Cref{thm-hp-X+}  it  remains to consider the case where condition~(\ref{negHA}) of \Cref{thm-sjf-to-sj} is true and condition~(\ref{negHMD}) is false. Thus we can assume that the query satisfies the following conditions\footnote{We can drop the condition $\vars(A)\cap \vars(B) \not\subseteq \key(A)$  because  $\key(A)\not\subseteq \key(B)$ and $\key(B)\subseteq \vars(A)$ together imply $\vars(A) \cap \vars(B) \not\subseteq \key(A)$ (and  we drop $\vars(A)\cap
      \vars(B) \not\subseteq \key(B)$ symmetrically).}:
\begin{align*}
  \key(A)\not\subseteq \key(B) \text{ and }
  \key(B)\not\subseteq \key(A) \text{ and }\\
  \key(A)\subseteq \vars(B) \text{ and } \key(B) \subseteq \vars(A)
\end{align*}

We call such queries {\bf \twowaydet}. Queries that are \twowaydet have special properties
that we will exploit to pinpoint the complexity of their consistent evaluation
problem. They are summarized in the following lemma.

\begin{lemma}
\label{lemma-key inclusion consequence}
Let $q$ be a \twowaydet query. Then for all
database $D$ and for all facts $a,b,c\in D$ suppose $D\models q(ab)$ then :
\begin{itemize}
\item  if $D \models q(ac)$ then $c\sim b$
\item  if $D \models q(cb)$ then $c\sim a$
\end{itemize}
\end{lemma}
\begin{proof}
   Assume $q=AB$ and $D\models q(ac)$. As $\key(B)\subseteq \vars(A)$ it follows that
   that $\bkey(c) = \bkey(b)$. The second claim is argued symmetrically using $\key(A) \subseteq \vars(B)$.
\end{proof}

In other words, within a repair, a fact can be part of at most two
solutions. Moreover, when a fact $e$ is part of two solutions of the repair,
the solutions must be of the form $q(de)$ and $q(ef)$. We then say that the
fact $e$ is {\bf branching} (with $d$ and $f$). If in addition $q(fd)$
holds then we say that $def$ is a {\bf triangle}, otherwise $def$ is a
{\bf fork}. The facts that can potentially be part of two solutions in a repair  play a crucial role in our proofs.  

When $q$ is \twowaydet, the complexity of $\certain(q)$ will depend on the
existence of a database, called \tripath, whose solutions to $q$ can be
arranged into a tree-like shape with one branching fact as specified next.

Let $d,e,f$ be three facts of a database $D$ such that $e$ is branching with $d,f$. Depending on the key
inclusion conditions of $def$, we define $\bg(e)$ as follows:

\begin{tabular}{l l l} if $\key(d) \subseteq \key(e)$ and $\key(f)
\not\subseteq \key(e)$& then &$\bg(e)=\bkey(d)$\\ if $\key(d) \not\subseteq
\key(e)$ and $\key(f) \subseteq \key(e)$& then &$\bg(e)=\bkey(f)$\\ if $\key(d)
\subseteq \key(f) \subseteq \key(e)$ & then &$\bg(e)=\bkey(d)$\\

 if $\key(f) \subseteq \key(d) \subseteq \key(e)$ & then &$\bg(e)=\bkey(f)$\\

 in all remaining cases & &$\bg(e)=\bkey(e)$
\end{tabular}

Note that $\bg(e)$ is well-defined because from
\Cref{lemma-key inclusion consequence} it follows that any other triple of the
form $d'ef'$ in $D$ such that $e$ is branching with $d',f'$ is such that
$d'\sim d$ and $f'\sim f$.  If $e$ is not branching then we define
$\bg(e) = \bkey(e)$. We also denote by $\g(e)$ the set of elements occurring in
the tuple $\bg(e)$. From the definition we always have
$\g(e)\subseteq \key(e)$.

A {\bf \tripath} of $q$ is a database $\Theta$ such that each block $B$
of $\Theta$ contains at most two facts, and all the blocks of $\Theta$ can be
arranged as a rooted tree with exactly two leaf blocks and satisfy the
following properties (see \Cref{figure-tripath-generic}). Let $s$ be the parent
function between the blocks of $\Theta$ giving its tree structure: if $B$
is a block of $\Theta$ then $s(B)$ denotes the parent block of $B$. Then:

\begin{figure*}[ht!]  
\begin{subfigure}[t]{\linewidth}
    \centering\includegraphics[width=0.7\linewidth]{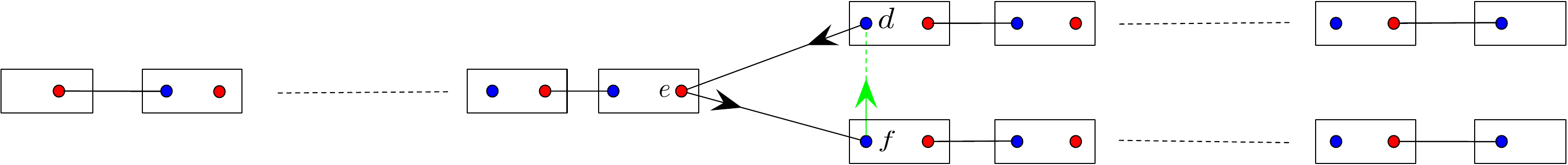}
    \caption{Generic structure of a \tripath. The rectangles denote blocks and
      in every block $B$, $a(B)$ is denoted by a red dot and $b(B)$ is denoted
      by a blue dot. The root block has only $a(B)$ and leaf blocks have only
      $b(B)$.  An undirected edge between two facts $s,t$ denotes that they form
      a solution $q\set{st}$. A directed edge from $s$ to $t$ denotes the
      solution $q(st)$.  The unique branching fact of the \tripath is denoted
      by $e$ which forms a solution with the facts $d$ and $e$ with
      $q(de) \land q(ef)$. $def$ is the center of the \tripath. If
      the green solution $q(fd)$ is present then we call it a
      triangle-\tripath, otherwise it is a fork-\tripath.  If the \tripath is
      not nice then there could be extra solutions to the query not depicted in
      the figure. The variable inclusion conditions are not depicted in the
      figure. }\label{figure-tripath-generic}
  \end{subfigure}
  
  \bigskip
  \begin{subfigure}[t]{\linewidth}
    \centering\includegraphics[scale=0.2]{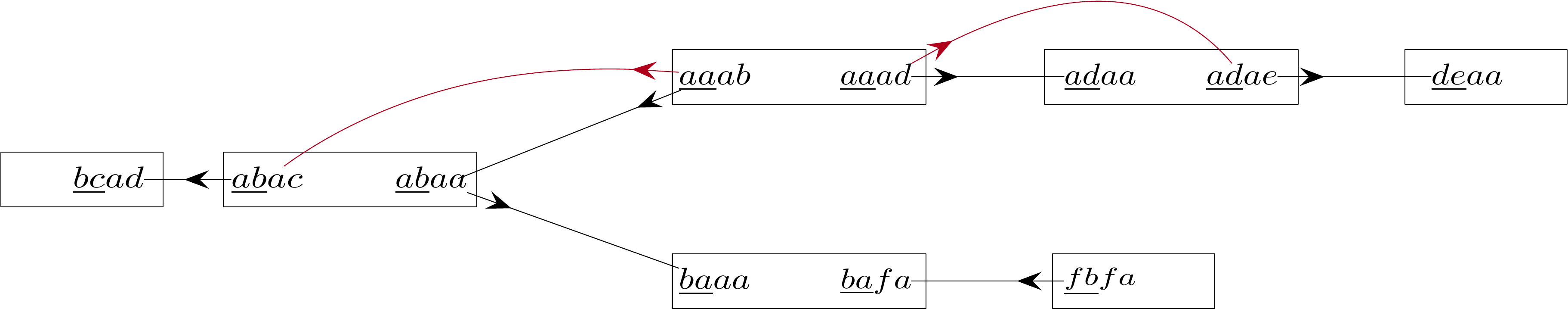}
    \caption{An instance of a \tripath for the query $q_2 = R(\underline{xu}~xy) \land R(\underline{uy}~xz)$. For the center of the given \tripath, $g(R(\underline{ab}aa))= \{a\}$. The facts in  the root and leaf blocks do not contain $a$ as a part of key. Note that there are extra solutions (in red) that are not enforced by the tripath. Hence this is not a nice-\tripath.}\label{figure-tripath-notnice}
  \end{subfigure}
  
  \bigskip
  \begin{subfigure}[t]{\linewidth}
    \centering\includegraphics[scale=0.2]{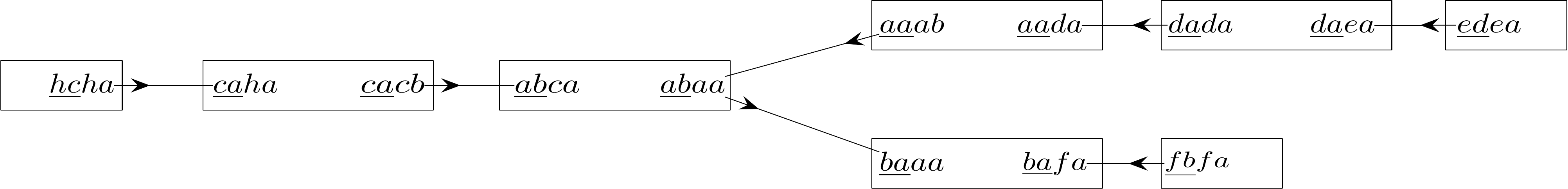}
    \caption{An instance of a nice-\tripath for the query $q_2 = R(\underline{xu}~xy) \land R(\underline{uy}~xz)$. We have the same center as in the previous case and $a$ does not belong to the key of the facts in root and leaf blocks.  Note that there are no extra solutions other than those enforced by the \tripath.}\label{figure-tripath-nice}
  \end{subfigure}
  \caption{Tripath illustrations }\label{figure-tripath}
  \Description{Tripath illustrations}
\end{figure*}

\begin{itemize}
\item There is exactly one block called the root block where the parent
  function $s$ is not defined and exactly two blocks, the leaf blocks, that
  have no children. Hence there is a unique block in $\Theta$, called the branching
  block, having two children.

\item Let $B$ be a block of $\Theta$. If $B$ is the root block, it contains exactly
  one fact denoted by $a(B)$. If $B$ is one of the leaf blocks then $B$ contains
  exactly one fact denoted by $b(B)$. In all other cases, $B$ contains exactly two
  facts denoted by $a(B)$ and $b(B)$.
  
\item Assume $B=s(B')$. Then $\Theta \models q\set{a(B)~b(B')}$.
  In particular, for
  the branching block $B$ we have $e= a(B)$ which is a branching fact with $d = b(B')$ and
  $f = b(B'')$, where $B'$ and $B''$ are the two blocks whose parent is the branching block $B$.
 We call the triple $def$ as the {\bf center} of the \tripath $\Theta$.

\item Let $B_0$, $B_1$, $B_2$ be respectively the root and leaves of $\Theta$ and
let $u_0=a(B_0)$, $u_1=b(B_1)$ and $u_2=b(B_2)$. Let $B$ be the center block of $\Theta$ and
$e=a(B)$. Then $\g(e) \not\subseteq \key(u_0)$ and $\g(e) \not\subseteq
\key(u_1)$ and $\g(e) \not\subseteq \key(u_2)$.
\end{itemize}

We say that a database $D$ contains a \tripath of $q$ if there exists
$\Theta\subseteq D$ such that $\Theta$ is a \tripath. A query $q$ admits a
\tripath if there is a database instance $D$ of $q$ that contains a \tripath.

A \tripath $\Theta$ is called a fork-\tripath if the center facts $def$ of
$\Theta$ forms a fork. If $def$ forms a triangle then $\Theta$ is called a
triangle-\tripath. 

The existence (or absence) of \tripath turns out to be the key in determining
the complexity of the consistent query answering problem of the \twowaydet
queries.

Notice that in the definition of \tripath we require the existence of some
solutions to $q$ (namely $q\set{a(B)b(B')}$ where $B$ is the parent block of $B'$) but we do not forbid the presence of other extra solutions. In order to
use the \tripath as a gadget for our lower bounds we need to remove those extra
solutions. To this end we introduce a normal
form for a \tripath that in particular requires no extra solutions, and show that
if a \tripath exists then there exists one in normal form.

For a \tripath $\Theta$, let  $B_0$, $B_1$, $B_2$ be respectively the root and leaves of $\Theta$ and
let $u_0=a(B_0)$, $u_1=b(B_1)$ and $u_2=b(B_2)$. We say that $\Theta$ is {\bf variable-nice} if
there exists $x\in \key(d), y\in \key(e)$ and $z\in \key(f)$ such that
$\set{x,y,z} \cap \big(\key(u_0) \cup \key(u_1) \cup
\key(u_2)\big)=\emptyset$. We say that a \tripath is
{\bf solution-nice} if
$q(\Theta)\subseteq \set{\set{ab} ~|~ a=a(B_i), b=b(B_j), s(B_i)=B_j} \cup \set{ \{fd\} }$.

The variable-nice property identifies three elements, each one from the key of the center facts $def$, such that the facts in the root and the leaf blocks do not contain these variables. These variables will be used in the encoding for proving \conp-hardness. The solution-nice property ensures that $q$ holds in $\Theta$ only where it must hold by
definition of being a \tripath, but nowhere else, with the only exception of
possibly $(fd)$, in which case $\Theta$ is a triangle-\tripath.

We say that a \tripath $\Theta$ is {\bf nice} if  the following holds:
\begin{itemize}
\item $\Theta$ is variable-nice
\item $\Theta$ is solution-nice
\item At least one of the
elements of $x,y,z$ (from being variable-nice), appears in the key of all facts
except  $u_0$, $u_1$ and $u_2$.
\item  Each of the keys of $u_0$,
$u_1$ and $u_2$ contains an element that does not occur in the key for any other
facts in $\Theta$. 
\end{itemize}

For instance the \tripath for $q_2$ depicted in \Cref{figure-tripath-notnice} is
not nice since it contains some extra solutions. However \Cref{figure-tripath-nice} depicts a nice \tripath for the
same query $q_2$. It turns out that niceness can be assumed without loss of
generality:

\begin{proposition}\label{prop-nice-tripath}
  Let $q$ be a \twowaydet query. If $q$ admits a fork-\tripath
  (triangle-\tripath) then $q$ admits a nice fork-\tripath
  (triangle-\tripath).
\end{proposition}

\begin{proof}[proof sketch]
Variable-niceness is achieved essentially by extending the branches of the \tripath  depending on how $\g(e)$ is defined. The construction of a solution-nice \tripath is more involved and is
done by induction on the number of extra solutions. Typically, if $q(\alpha\beta)$ is an extra solution we will
replace the fact $\alpha$ so that this extra solution is removed and add new blocks to the \tripath so that all other properties are satisfied. This can only work if
$\alpha$ is not part of the center of the \tripath. When $\alpha$ is part of
the center, it turns out that $\beta$ can not be part of the center. We then
argue by symmetry using the block of $\beta$. The last two conditions are again simple to achieve.
\end{proof}

We will show in \Cref{sec-tripath or chase} that if a query $q$ does not admit
a \tripath then $\certain(q)$ can be solved in polynomial time using the greedy
fixpoint algorithm of \Cref{sec-cert-algo}. If a query $q$ admits a
fork-\tripath we will show in \Cref{sec-fork tripath and conp} that
$\certain(q)$ is \conp-complete. If a query $q$ does not admit a fork-\tripath
but admits a triangle-\tripath we will
show in \Cref{sec-triangle tripath} that $\certain(q)$ can be solved in
polynomial time, using a combination of the fixpoint algorithm of
\Cref{sec-cert-algo} and bipartite matching.

%% file: tripath-or-chase.tex
The main goal of this section is to prove that for every \twowaydet query $q$,
if $q$ does not admit a \tripath then $\certain(q)$ is in \ptime.  There are
many \twowaydet queries that have no \tripath. For instance the query
$q_5=R(\underline{x}~yx)R(\underline{y}~xu)$ does not admit a \tripath because
any three facts $d,e,f$ such that $q_5(de)\land q_5(ef)$ holds are such that
two of them must be in the same block. This implies that we cannot have a center
for the \tripath, so $q_5$ does not admit a \tripath. Hence $\certain(q_5)$ is
in \ptime which follows from the next theorem.

\begin{theorem}\label{corollary-no tripath implies cqk}
  Let $q$ be a \twowaydet query. If $q$ does not admit a \tripath
  then $\certain(q)$ is in \ptime.
\end{theorem}

In fact we show that $\certain(q)$ can be solved using the greedy fixpoint
algorithm $\Cqk$ defined in \Cref{sec-cert-algo} for $k=2^{2\kappa+1}+\kappa -
1$ where $\kappa = l^l$ (recall that $l$ is the number of key positions in the relation
$R$ under consideration)\footnote{Note that the constant $k$ directly results from the proof technique and is not intended to be optimal.}.

\begin{proposition}\label{lemma-k minimal or tripath}
  Let $q$ be a \twowaydet query and let $k = 
  2^{2\kappa+1}+\kappa - 1$ and let $D$ be a database. If
  $D$ does not admit a \tripath of $q$ then  $D\in\certain(q)$ iff $D\in \Cqk$.
\end{proposition}

For any database $D$ and repair $r$ of $D$ and $a\in r$ let
$\rkey(a,r) = \{ c\mid c\in r$ and $\key(c)\subseteq \key(a)\}$. 
To prove \Cref{lemma-k minimal or tripath} we build on the following lemma
which resembles \Cref{lemma-hp-X+} but requires a more involved technical proof.

\begin{lemma}\label{lem-k-minimal}
  Let $k\ge \kappa$ and $q=AB$ be a query that is \twowaydet. Let $D$ be a
  database that does not admit a \tripath. Then for every repair $r$ of $D$
  such that $r\models q\set{ab}$, and for all $K\subseteq r$ such that
  $\rkey(a) \subseteq K$ and $|K| \le k$, one of the following conditions hold:
\begin{enumerate}
\item $K\in \Delta_k(q,D)$
\item There exists a repair $r'$ such that $K\subseteq r'$ and $q(r') \subsetneq q(r)$.
\end{enumerate}

\end{lemma}

\begin{proof}[Proof sketch]

 Assume that $D$
does not admit a \tripath. Consider  a repair $r$ of $D$, two
facts $a,b$ of $r$ such that $r\models q\set{ab}$ and let $K$ be a set of facts
containing $\rkey(a,r)$. We need to show that if $K$ is not in $\Delta_k(q,D)$
then there is a new repair with strictly less solutions.  We therefore need to
remove at least one solution from $r$ and we have an obvious candidate as
$r \models q\set{ab}$. We then use the definition of the algorithm $\cert_k$
and from the fact that $K\not\in\Delta_k(q,D)$ we know that in the block of $b$
there is a fact $b'$ that can not form a $k$-set when combined with facts from
$K$. In particular $b'$ and $a$ do not form a solution to $q$.  Let $r' = r[b\to b']$.

Now if $b'$ is not a part of any solution in $r'$ then $r'$ is the desired
repair. Otherwise we can repeat the above argument with the facts that are
making $q$ true when combined with $b'$. There are at most two such facts and
we need to consider them both, one after the other. The goal is to repeat the
process above until all newly created solutions are removed from the working
repair. When doing so we visit blocks of $D$, selecting two facts in each such
block, the one that makes the query true with a previously selected fact, and
the one we obtain from the non-membership to $\Delta_k(q,D)$.  If we can
enforce that we never visit a block twice we are done because the database
being finite, eventually all the selected facts will not participate in a
solution to $q$ in the current repair.  In order to do this we keep in memory
(the set $K$ initially) all the facts of the current repair that can
potentially form a new solution. This ensures that we
will never make the query true with a fact in a previously visited block. The
difficulty is to ensure that the size of the memory remains bounded by
$k$. This is achieved by requiring key inclusion between two consecutively
selected facts, otherwise we stop and put a flag. If we get two flags we argue
that we can extract a \tripath which is a contradiction. If not we can show
that the memory remains bounded.
\end{proof}

We conclude this section with a sketch of the of proof of \Cref{lemma-k minimal
  or tripath}. We assume $D \models \certain(q)$ and show that $D \models \Cqk$. Let $r$ be a
  repair of $D$ that contains a minimal number of solutions. For any set
of facts $K\subseteq r$ and $K'\subseteq D$, denote $K'\sim K$ if there is a
bijection $f: K\to K'$ where $f(a) \sim a$ and let $r[K\to K']$ be the new
repair obtained by replacing the facts of $K$ in $r$ by the facts of $K'$. 
 
  Since $D\models \certain(q)$ there exists $a,b\in r$ such that
  $r\models q\set{ab}$. Let $K = \rkey(a,r)$, clearly $|K|\le \kappa$.  It
  suffices to show that for all $K'\sim K$, $K'\in \Delta_k(q,D)$.  If this is
  not the case for some $K'$, let $r' = r[K\to K']$.  As $r$ is minimal, we can
  assume there are facts $c\in K'$ and $d\in r'$ such that $q\set{cd}$.

  Notice that $\rkey(c,r') \subseteq K'$. As $K'\not\in \Delta_k(q,D)$, by
  \Cref{lem-k-minimal}, there exists a repair $r''$ such that $K'\subseteq r''$
  and $q(r'') \subsetneq q(r')$. Repeating this argument eventually yields a
  repair contradicting the minimality of $r$.


%% file: tripath-conp.tex
In this section we prove that if a query that is \twowaydet admits a fork-\tripath, then $\certain(q)$ is \conp-hard.  We have already
seen that the query $q_2= R(\underline{xu}~xy) \land R(\underline{uy}~xz)$
admits a fork-\tripath (associated fork-\tripath are depicted in
\Cref{figure-tripath} part (b) and (c)).  
The fact that $\certain(q_2)$ is \conp-hard is a consequence of the following result.

\begin{theorem}
\label{theorem-conp hardness}
Let $q$ be a query that is \twowaydet. If $q$ admits a
fork-\tripath, then \certain(q) is \conp-complete.
\end{theorem}

The remaining part of this section is a proof of \Cref{theorem-conp hardness}.
In view of Proposition~\ref{prop-nice-tripath} we can assume that $q$ has a
nice fork-\tripath $\Theta$. Let $x,y,z$ be the elements of $\Theta$ witnessing the
variable-niceness of $\Theta$ and let $u,v,w$ be the fresh new elements occurring only
in the keys of the head and tails of $\Theta$. Note that $x,y,z$ need not be distinct. For any elements
$\alpha_x,\alpha_y,\alpha_z,\alpha_u,\alpha_v\alpha_w$, we denote by
$\Theta[\alpha_x,\alpha_y,\alpha_z,\alpha_u,\alpha_v,\alpha_w]$ the database constructed from $\Theta$ by
replacing each of $x,y,z,u,v,w$ by $\alpha_x,\alpha_y,\alpha_z,\alpha_u,\alpha_v,\alpha_w$ respectively, where $\alpha_x = \alpha_y$ iff $x=y$; $\alpha_y = \alpha_z$  iff $y=z$ and so on.

We use a reduction from $3$-SAT where every variable occurs at most
$3$ times.  Let $\phi$ be such a formula. Let $V_2$ be the variables of $\phi$
that occur exactly two times and $V_3$ be the variables of $\phi$ that occur exactly three times. Without loss of generality we can assume that each variable $p$ of
$\phi$ occur at least once positively and at least once negatively.
The construction is illustrated in \Cref{fig:lower-bound}.

\begin{figure*}[!ht]
  \centering\includegraphics[width=\linewidth]{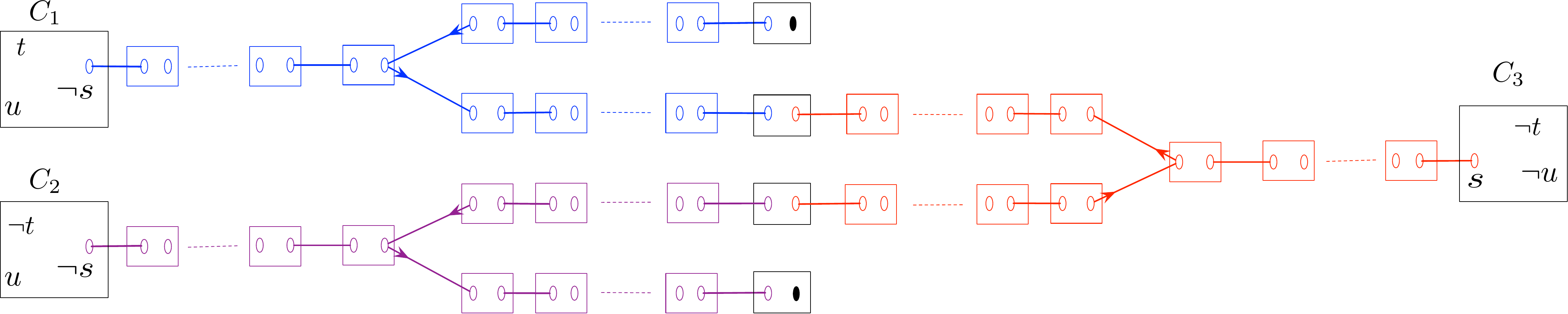}
  \Description{CoNP gadget illustration}
  \caption{Consider the SAT formula
    $(\neg s \lor t \lor u) \land (\neg s \lor \neg t \lor u) \land (s \lor
    \neg t \lor \neg u)$. Each clause has a corresponding block denoted by
    $C_1, C_2$ and $C_3$ respectively. Each such block has three facts
    corresponding to the literals of the clause. The figure illustrates the gadget
    for the variable $s$. Similar construction is also done for the variables
    $t$ and $u$. Note that because the \tripath is solution-nice, if a repair $r$ makes the query false and picks
    $\neg s$ from $C_1$ and/or $C_2$ then it cannot have $s$ from
    $C_3$. Conversely if $r$ contains $s$ from $C_3$ then it cannot have
    $\neg s$ from both $C_1$ and $C_2$. The variable-niceness of the \tripath
    provides the necessary variables to encode the clause and the literals.}
  \label{fig:lower-bound}
\end{figure*}

\paragraph{The database.} 
Let $l \in V_3$. By our assumption, $l$ (or $\neg l$) occurs once positively - let $C[l]$ be
this clause - and twice negatively - let $C_1[l]$, $C_2[l]$ be the two corresponding clauses.

Let $D[l]$ be the database consisting of the union of:\\
$\Theta_{l,C}=\Theta[\tup{C,l}_x,\tup{C,l}_y,\tup{C,l}_z,C,\tup{C,C_2,l},\tup{C,C_1,l}]$\\
$\Theta_{l,C_1}=\Theta[\tup{C_1,l}_x,\tup{C_1,l}_y,\tup{C_1,l}_z,C_1,\tup{C_1,C_1,l},\tup{C,C_1,l}]$ and\\
$\Theta_{l,C_2}=\Theta[\tup{C_2,l}_x,\tup{C_2,l}_y,\tup{C_2,l}_z,C_2,\tup{C,C_2,l},\tup{C_2,C_2,l}]$.

A few remarks about $D[l]$. The right leaf of $\Theta_{l,C}$ has the
same key as the right leaf of $\Theta_{l,C_1}$ while the left leaf of $\Theta_{l,C}$ has
the same key as the left leaf of $\Theta_{l,C_2}$. For the remaining blocks, the
union is a disjoint union (because they contain $x,y$ or $z$ in $\Theta$, hence the
element $\tup{C,l}_x,\tup{C,l}_y$ or $\tup{C,l}_z$ in $\Theta_{C,l}$ and so on). In particular all blocks have size two and
each fact make the query true with a fact in a adjacent block.

Let now $l \in V_2$. By our assumption, $l$ (or $\neg l$) occur once positively - let $C[l]$ be
this clause - and once negatively - let $C'[l]$ be the corresponding clause.

Let $D[l]$ be the database consisting of the union of\\
$\Theta_{l,C}=\Theta[\tup{C,l}_x,\tup{C,l}_y,\tup{C,l}_z,C,\tup{C,C,l},\tup{C,C',l}]$ and\\
$\Theta_{l,C'}=\Theta[\tup{C',l}_x,\tup{C',l}_y,\tup{C',l}_z,C',\tup{C',C',l},\tup{C,C',l}]$.

For the given $3$-SAT formula $\phi$, define the corresponding database
instance $D[\phi]$ as $\bigcup_{l\in\phi} D[l]$. Further, for every block $B$ in $D[\phi]$ if $B$ contains only one fact, then add a fresh fact in the block of $B$ that does not form a solution with any other facts of $D[\phi]$ (such a fact can always be defined for any block).

A few remarks about $D[\phi]$. Notice that by construction, every block of $D[\phi]$ has at least two facts. If $l$ and $l'$ occur in the same
clause $C$ then the roots of $\Theta_{l,C}$ and $\Theta_{l',C}$ have the same keys. We
call these heads the block of $C$ in the sequel. For
all other blocks, the union of the $D[l]$ is a disjoint union of blocks as they
all contain an element annotated with $l$ in their key.
Consider now a pair of fact $a,b$ such that $D[\phi] \models q\set{ab}$.
As the key of each element of $D[\phi]$ is annotated by either $C$ or $l$, $a$
and $b$ must belongs to the same $\Theta_{C,l}$ because $q$ is \twowaydet. Hence $a,b$ must be homomorphic copies of $a',b'$
in $\Theta$ such that $\Theta\models q(a'b')$. As $\Theta$ is solution-nice, $a',b'$ must be in
consecutive blocks in $\Theta$.

The following lemma concludes the proof of \Cref{theorem-conp hardness}.

\begin{lemma}
\label{lemma-conp proof}
Let $\phi$ be a 3-sat formula where every variable occurs at most three
times. $\phi$ is satisfiable iff $D[\phi] \not\models \certain(q)$.
\end{lemma}

%% file: triangle-tripath.tex
It remains to consider queries that admit at least one triangle-\tripath but no
fork-\tripath. This is for instance the case for the query
$q_6=R(\underline{x}~yz) \land R(\underline{z}~xy)$ since $q_6$ does not have a
fork-\tripath as all branching facts for $q_6$ form a triangle. However it is easy
to construct a triangle-\tripath for $q_6$.  A more challenging example is
{\small $q_7 = R(\underline{x_1x_2x_3~y_1y_1y_2y_3~z_1z_2z_3}~z_4z_4z_4z_4) \land
  R(\underline{x_3x_1x_2~y_3y_1y_1y_2~z_2z_3z_4}~z_1z_2z_3z_4)$}.
It is not immediate to construct a triangle-\tripath for $q_7$ and even less
immediate to show that $q_7$ admits no fork-\tripath. This is left as a useful
exercise to the reader.

We show
in this section that for such queries, certain answers can be computed in polynomial
time. However, the following result shows that the greedy fixpoint
algorithm of \Cref{sec-cert-algo} does not work for such queries.

\begin{theorem}\label{thm-triangle-tripath-lower-bound}
Let $q$ be a \twowaydet query admitting a triangle-\tripath. Then for all $k$,
$\certain(q) \neq \Cqk$. 
\end{theorem}

The proof is essentially a reduction to the query
$q_6$ for which it is shown in \cite{ICDTJournal}
that $\certain(q_6)$ can not be solved using $\cert_k(q_6)$, for all
$k$.

\subsection{Bipartite matching algorithm}\label{section-bipartite}

Since the algorithm $\Cqk$ does not work, we need a different polynomial time
algorithm to handle these queries. We use an algorithm based on bipartite
matching, slightly extending the one introduced in \cite{ICDTJournal} for self-join-free queries.

Note that since we only consider queries with two atoms, for every database
$D$, it is convenient to describe the set of solutions to a query $q$ as a
graph. 
We define the solution graph of $D$, denoted by $G(D,q)$ to be an undirected
graph whose vertices are the facts of $D$ and there is an edge between two
facts $a$ and $b$ in $G(D,q)$ iff $D \models q\set{ab}$.  A connected component
$C$ of $G(D,q)$ is called a {\bf quasi-clique} if for all facts
$a,b\in C$ such that $a\not\sim b$, $\set{a,b}$ is an edge in $G(D,q)$.

For an arbitrary database $D$, and each fact $a$ of $D$, $\clique(a)$ is
defined as follows : if $C$ is the connected component of $G(D,q)$ containing
$a$ and $C$ is a quasi-clique $\clique(a) = C$, otherwise
$\clique(a) =\set{a}$.

On input $D$, \matching first computes $G(D,q)$ and its connected components,
and then creates a bipartite graph $H(D,q) = (V_1 \cup V_2, E)$, where $V_1$ is
the set of blocks of $D$ and $V_2 =\set{clique(a)~|~a \in D}$. Further
$(v_1,v_2) \in E$ iff the block $v_1$ contains a fact $a$ which is in $v_2$ and
such that $D\not\models q(a,a)$.
Note that constructing
$G(D,q)$ and $H(D,q)$ can be achieved in polynomial time.  Finally the
algorithm outputs `yes' iff there is a bipartite matching of $H(D,q)$ that
saturates $V_1$. In this case we write $D \models \matching$.

This can be checked in
\ptime\cite{DBLP:journals/siamcomp/HopcroftK73}. 
We now show that $\neg \matching$ 
is always an under-approximation of $\certain(q)$.

 \begin{proposition}\label{prop-matching-under-approx}
   Let $q$ be a \twowaydet query and $D$ be a database. Then
   $D \models \neg\matching$ implies  $D\models \certain(q)$. 
 \end{proposition}

 A database $D$ is called a {\bf clique-database} for $q$ if every
 connected component $C$ of the graph $G(D,q)$ is a quasi-clique.
 As soon as the input database is a clique-database for $q$, $\neg\matching$ correctly computes $\certain(q)$ .
 
 \begin{proposition}\label{thm-triangle-query-ptime}
   Let $q$ be a \twowaydet query and $D$ be a clique-database for $q$. Then
   $D \models \neg\matching$ iff $D\models \certain(q)$. Therefore checking whether $D\models \certain(q)$ is in \ptime.
 \end{proposition}

This already gives the complexity of $\certain(q)$ for some queries that do not admit fork-\tripath (but possibly admit triangle-\tripath). A query $q$ is said to be a {\bf clique-query} if
every database $D$ is a clique-database for $q$.  For instance, the query
$q_6$ is a clique-query as the
solution graph of any database is a clique-database. From \cref{thm-triangle-query-ptime} the following theorem follows.

\begin{theorem}\label{cor-clique-query is ptime}
Let $q$ be a \twowaydet query. If $q$ is a  clique-query then $\certain(q)=\neg\matching$, and thus $\certain(q)$ is in \ptime.
\end{theorem}

\subsection{Combining matching-based and
  greedy fixpoint algorithms}\label{sec-combine-mathing-and-cqk}

In this section we prove that certain answers to a \twowaydet query $q$ which does not admit a fork-\tripath are computed by a combination of the polynomial time algorithms $\matching$ and $\Cqk$, thus completing the dichotomy classification (Recall that $\kappa = l^l$ where $l$ is the number of key positions).

\begin{theorem}
\label{thm-no fork triangle path}
Let $q$ be a query that is \twowaydet. If $q$ does not admit a fork-\tripath
then $\certain(q) = \Cqk \vee \neg\matching$, for $k=2^{2\kappa+1}+\kappa - 1$.
Thus $\certain(q)$ is in polynomial time.
\end{theorem}

The key to prove this theorem is the following proposition
which proves that the database can be partitioned into components, such that on each component at least one of the two polynomial time algorithms is correct.

\begin{proposition}
  \label{prop-component-partition} Let $q$ be a \twowaydet query that does not admit a fork-\tripath 
  and let $D$ be a database.  There exists a partition $C_1,  C_2, \dots C_n$ of $D$ having all of the following properties :
  \begin{enumerate}
\item for all $i$, $C_i$ does not contain a \tripath or $C_i$ is a clique-database for $q$.\label{item-connected-property}
  \item $D\models \certain(q)$ iff there exists some $i$ such that $C_i\models \certain(q)$.\label{item-certain-partition}
  \item For all $k$, if $C_i\models \Cqk$ for some $i$, then $D\models \Cqk$.\label{item-cqk-partition}
  \item If $D\models \matching$ then for all $i$  $C_i\models \matching$. \label{item-matching-partition}
  \end{enumerate}
\end{proposition}

\begin{proof}[Proof sketch]
  The partition of the database is obtained using the following equivalence
  relation. Two blocks $B, B'$ of a database $D$ are said to be {\bf
    $q$-connected} if $(B, B')$ belongs to the reflexive symmetric transitive
  closure of
  $\set{(B_1,B_2)~|~ \exists a \in B_1, b\in B_2 \textrm{ such that }D\models
    q\set{ab}}$.  The main difficulty is to show that each $q$-connected
  component satisfies $(1)$ (the remaining properties are easy to show). This
  is proved by showing that if a $q$-connected component contains both a
  triangle-\tripath and a fork, then $q$ also admits a fork-\tripath. This is
  the main technical contribution.  Altogether, if $q$ does not admit a
  fork-tripath then each $q$-connected component is either a clique-database or
  contains no \tripath at all.
\end{proof}

\begin{proof}[Proof of \Cref{thm-no fork triangle path}]
We  assume \Cref{prop-component-partition} and prove the theorem. For an input database $D$, if $D \not\models\certain(q)$ then, by~\Cref{prop-matching-under-approx}, $D \models \matching$; moreover $D \not\models \Cqk$, as $\Cqk$ too is always an under-approximation of $\certain(q)$.

Assume now $D \models\certain(q)$ and $D \models \matching$, we show that $D\models \Cqk$.
Consider the partition $C_1, \dots C_n$ of $D$ given by~\Cref{prop-component-partition}.
Since $D \models\certain(q)$ there exists a $C_j$ such that $C_j \models \certain(q)$; moreover $C_i\models \matching$ for all $i$. In particular $C_j \models \matching$ and therefore on $C_j$ the $\neg\matching$ algorithm does not compute certain answers. Then by~\Cref{thm-triangle-query-ptime}, $C_j$ is not a clique-database for $q$.
Now by~\Cref{prop-component-partition}, $C_j$ admits no \tripath, and therefore by \Cref{lemma-k minimal or tripath}, $C_j \models \Cqk$.
Then, again by~\Cref{prop-component-partition}, $D \models \Cqk$.
\end{proof}


%% file: conclusion.tex
We have proved the dichotomy conjecture on consistent query answering for queries with two atoms.
The conditions we provided for separating the polynomial time case from the \conp-hard case can be shown to be decidable. Indeed one can show that if a fork-\tripath exists
then there exists one of exponential size. However, it is likely that there are more efficient decision procedures than testing the existence
of a \tripath.

We also obtained a (decidable) characterization of the two-atom queries whose certain answers are computable using the greedy fixpoint algorithm of
\Cref{sec-cert-algo}. The current characterization assumes $\ptime \neq \conp$
but we believe that if a query $q$ admits a fork-tripath then
$\certain(q_6)$ reduces to $\certain(q)$ as in the triangle-\tripath case,
hence proving the characterization without any complexity assumption. We postpone
this for the journal version of this paper.

The dichotomy conjecture for all conjunctive queries remains a
challenging problem. 
A new challenge that this paper poses is that of characterizing all conjunctive queries whose certain answers are computable by the greedy fixpoint algorithm of \Cref{sec-cert-algo}.
We believe this is a worthwhile question, given the simplicity
of this algorithm.

Along the lines of \cite{ICDTJournal} we can also attempt to characterize the queries solvable in \fo. We conjecture that this class coincides with those queries that can be solved using the fixpoint algorithm where the fixpoint terminates after a bounded number of computations, irrespective of the size of the database.
